\documentclass{llncs}

\usepackage{a4wide}
\usepackage[utf8]{inputenc}
\usepackage{amsmath}
\usepackage{amsfonts}
\usepackage{graphicx}
\usepackage{hyphenat}
\usepackage{enumitem}
\usepackage{amsfonts}
\usepackage{algorithmicx}
\usepackage{xspace}
\usepackage{array}
\usepackage{todonotes}

\newcommand{\NP}{{$\mathsf{NP}$}\xspace}

\newcommand{\PP}{{$\mathsf{P}$}\xspace} 
\newcommand{\SKA}{{\em SKA}\xspace}
\newcommand{\NSKA}{{\em NSKA}\xspace}
\newcommand{\NISKA}{{\em (N)SKA}\xspace}
\newcommand{\SST}{{$SST$}\xspace}
\newcommand{\SMST}{{$SMST$}\xspace}
\newcommand{\OISMST}{{$01\text{-}SMST$}\xspace}
\newcommand{\POISMST}{{$\cap\text{-}01\text{-}SMST$}\xspace}

\newtheorem{observation}[theorem]{Observation}
\newtheorem{fact}[theorem]{Fact}

\newcommand{\R}{\mathbb{R}}   

\begin{document}

\title{On the Simultaneous Minimum Spanning Trees Problem}

\author{
Mat\v ej Kone\v cn\'y,
Stanislav Ku\v cera,
Jana Novotn\'a, 
Jakub Pek\'{a}rek, 
Martin Smol\'ik,
Jakub T\v etek,
Martin T\"opfer}
\institute{Department of Applied Mathematics, Faculty of Mathematics and Physics, Charles
University, Prague, Czech Republic,\\
\email{matejkon@gmail.com, stanislav.kucera@outlook.com, janca@kam.mff.cuni.cz,
edalegos@gmail.com,
smolik.ma@gmail.com,
j.tetek@gmail.com,
mtopfer@gmail.com}
}
\maketitle

\begin{abstract}


Simultaneous Embedding with Fixed Edges (SEFE) \cite{b:sefe} is a problem where given $k$ planar graphs we ask whether they can be simultaneously embedded so that the embedding of each graph is planar and common edges are drawn the same. Problems of SEFE type have inspired questions of Simultaneous Geometrical Representations and further derivations. Based on this motivation we investigate the generalization of the simultaneous paradigm on the classical combinatorial problem of minimum spanning trees. Given $k$ graphs with weighted edges, such that they have a common intersection, are there minimum spanning trees of the respective graphs such that they agree on the intersection? We show that the unweighted case is polynomial-time solvable while the weighted case is only polynomial-time solvable for $k=2$ and it is \NP-complete for $k\geq 3$.


\end{abstract}

\section{Introduction} 
The problem of finding a minimum spanning tree is one of the most important and most well-studied problems in graph algorithms. We consider a variant of this problem inspired by the following motivation. 

In a Sunflower land, there is a capital city and several smaller cities around it. In the past, there was a telecommunication company based in the capital city, but it is now bankrupt. The inhabitants of each of the small cities want to establish their own telecommunication company that would connect all of the houses in their city as well as all of the houses in the capital. The representatives of each city meet to coordinate their soon-to-be networks so that they all agree on the capital and can split the cost of covering the capital evenly. However, all of the companies are so afraid of bankruptcy that none of them would accept a solution that would cost them a single dollar more than necessary. Is it always possible to plan all of the networks so that all of the companies reach their goal simultaneously while each of the individual costs is minimized? How hard is it to find such a plan, if it exists, or recognize that it does not exist?

\begin{problem}[Simultaneous Minimum Spanning Trees]
Let $k$ be a positive integer and let $G_1=(V_1, E_1),$ $G_2=(V_2, E_2), \ldots, G_k=(V_k, E_k)$ be graphs and $w$ a non-negative weight function of all of their edges ($w:{\bigcup}_{i=1}^{k} E_i \rightarrow \R^{+}_0$) such that there is a graph $\bar{G}$ satisfying that $\bar{G} = G_i[\bar{V}]$ for any $i$ from 1 to $k$, where $\bar{V} = V_i \cap V_j$ for any $i \neq j$ from 1 to $k$ (i.e. the graphs together form a ``sunflower'' shape with no lateral edges). Find minimum spanning trees $T_i\subseteq G_i$, such that they all coincide on $\bar{G}$, or answer $NO$ if there are no such spanning trees. We shall abbreviate this problem as \SMST.
\end{problem}

Note that the $T_i$'s do not have to induce a spanning tree on $\bar{G}$, nor does the union of $T_i$'s have to be acyclic on the union of all of the $G_i$'s. Indeed both of these situations necessarily happen in solutions of some instances of the \SMST problem. Unlike the minimum spanning tree problem, the \SMST problem does not always have a solution.

As an example, let $G_1$ be a triangle $xzy$, let $G_2$ be a triangle $xwy$ and let $xy$ be the heavies edge. Although $G_1 \cap G_2$ induces a connected graph (edge $xy$), we have a unique solution $\{xz,zy,yw,wx\}$ which is not connected on $G_1 \cap G_2$ and is not acyclic on $G_1 \cup G_2$. Furthermore, if we remove any light edge, e.g. $xz$, then there is no solution. 

We show that \SMST is an \NP-complete problem already for a fairly small number of graphs (more than 2) and even when limited to simplified instances. We present a scheme that allows us to solve any \SMST for two graphs in polynomial time using a tandem of reductions and multiple runs of matroid intersection algorithm. 

\subsection{Preliminaries}
The problem of finding a minimum spanning tree for a single graph has been studied thoroughly since Bor\r uvka \cite{b:boruvka}, Jarn\' ik-Prim \cite{b:jarnik}\cite{b:prim} and Kruskal \cite{b:kruskal}. See \cite{b:mst} for more details. Currently, the optimal algorithm is known \cite{b:mstopt}, but its asymptotics is still an open problem. 

We do not distinguish instances where the input graph is connected from instances where it is disconnected. The inclusion of disconnected instances is natural as many constructions work just as well under such circumstances. Furthermore, usual incremental and iterative approaches typically work on subsets of the input graph and it is therefore not strictly clear whether they maintain a spanning tree or a spanning forest. For convenience, we define the usual term {\em spanning tree} as a maximal acyclic subgraph. In doing so we include the disconnected case, where the more proper term would be {\em spanning forest}. 

We focus mainly on the Kruskal's algorithm and use its known properties. Kruskal's algorithm starts by sorting the edges in a non-decreasing order (by weight) or obtains the edges in a non-decreasing order on input. Then it processes all the edges one by one in sorted order while greedily maintaining maximum acyclic subgraph which we refer to as {\em partial spanning tree}. 

\begin{definition}
Consider the run of Kruskal's algorithm. A {\em stage} is a collection of steps in which the algorithm processes edges of the same weight. 
\end{definition}

\begin{fact}\label{fact:kruskal}
Let $G = (V,E,w)$ be a graph with weighted edges. Then all of the following holds for Kruskal's algorithm applied to graph $G$ and a non-decreasing order of edges $\pi$: 
\begin{itemize}
    \item Kruskal's algorithm is complete (finishes) and correct (answers correctly) for any non-decreasing $\pi$, although the created spanning trees might be different. 
    \item Let $T$ be a minimum spanning tree of $G$ and let $\pi_T$ be the non-decreasing order such that all edges from $T$ are ordered before all edges of the same weight that are not from $T$. Then Kruskal's algorithm using $\pi_T$ outputs exactly $T$. 
    \item After every stage, components of the partial spanning tree span across the same vertices for all non-decreasing $\pi$. 
    \item Edges added to the partial spanning tree in each stage depend only on their ordering, not on the edges chosen in the previous stages. 
    \item Kruskal's algorithm accesses $\pi$ in a read-once fashion, accepting or refusing each edge before accessing the next one. 
\end{itemize}
\end{fact}
\section{Simultaneous Kruskal's algorithm}

Consider a \SMST task for a given $k$ and graphs $G_1, G_2, \dots, G_k$. Let us denote the union of all $G_i$s as $G$ and their intersection as $\bar{G}$. Suppose we order all the edges of $G$ in a non-decreasing order $\pi$ which we call a {\em universal order} and denote $\pi[E(G_i)]$ the restrictions of $\pi$ to edges in $G_i$ for every $i$. For a set of edges $F$, we also say that a universal order is {\em F-preferring} if all the edges from $F$ are ordered before any other edges of the same weight. 

Consider the following construction. First, we fix an arbitrary non-decreasing universal order $\pi$. We simulate $k$ independent instances of Kruskal's algorithm, $K_1,K_2,...,K_k$ where the job of each $K_i$ is to find a minimum spanning tree $T_i$ of $G_i$ using the order $\pi[E(G_i)]$, not considering the other instances. In parallel with the instances of the Kruskal's algorithm we try to incrementally build a simultaneous minimum spanning tree. 

In the beginning, we start with an empty simultaneous spanning tree $T$ and process all the edges one by one according to the universal order. We present each edge $e$ to all instances $K_i$ such that $e \in G_i$. If we assume a sunflower intersection, we can rephrase this in the following way: if $e \in \bar{G}$ then we present $e$ to all instances and if $e \notin \bar{G}$ then $e \in G_j$ for some unique $j$ and we present $e$ only to one instance $K_j$. If every invoked $K_i$ adds $e$ to its local $T_i$, we also add $e$ to $T$. If every invoked instance $K_i$ refuses to add $e$ to its local $T_i$, we also throw $e$ away. If the invoked instances do not agree, we fail. If the algorithm processes all edges without failing, we output $T$ as a solution. 

We call this construction {\em Simultaneous Kruskal's algorithm} or \SKA in short. There are two natural versions of the \SKA. If \SKA expects the universal order $\pi$ on input, then it is a deterministic algorithm. Alternatively, \SKA may be formulated as a non-deterministic algorithm which guesses a correct universal order which avoids failure (if any such order exists), then we speak of a non-deterministic simultaneous Kruskal's algorithm or \NSKA in short. We naturally extend the definition of a stage from Kruskal's algorithm to the \NISKA as the collection of steps in which the algorithm processes edges of the same weight. 

\begin{lemma}\label{lem:ska}
Let $I$ be an instance of the \SMST problem. Then all of the following holds for simultaneous Kruskal's algorithm: 

\begin{itemize}
    \item \NSKA is complete (finishes) and correct (answers correctly).  
    \item Let $T$ be a solution of $I$ and a let $\pi_T$ be a $T$-preferring universal order. Then \SKA using $\pi_T$ outputs exactly $T$. 
    \item After every successful stage of \SKA and \NSKA, components of the partial simultaneous spanning tree after restriction to any $G_i$ span across the same vertices for all choices of universal order $\pi$. 
    \item Edges added to the partial simultaneous spanning tree in each stage depend only on their ordering, not on the edges chosen in the previous stages. 
    \item \SKA accesses $\pi$ in a read-once fashion, accepting or refusing each edge before accessing the next one. 
\end{itemize}
\end{lemma}
\begin{proof}
Let us first prove the second point. Suppose we run the \SKA using the $T$-preferring universal order. Let us analyze the behavior of an arbitrary $K_i$. Let $T_i$ denote the restriction of $T$ to $G_i$. By definition $T_i$ is a minimum spanning tree of $G_i$ and $\pi[G_i]$ is a $T_i$-preferring order. From the properties of the Kruskal's algorithm (fact \ref{fact:kruskal}) we know that $K_i$ constructs exactly $T_i$. Since every $K_i$ would construct exactly $T_i$ should it run on its own, we observe that all the invoked instances $K_j$ accept each edge if and only if it belongs to $T$, and the whole algorithm never fails. At the end of the computation the algorithm gives exactly $T$ as a solution. 

To prove correctness, let us first suppose that the \NSKA terminates with success. Then the set $T$ on output is a union of the local spanning trees from all $K_i$ algorithms. Since each algorithm $K_j$ processes all the edges from $G_j$ in a non-decreasing order of weight and Kruskal's algorithm is sound, each of the local spanning trees is a minimum spanning tree. Thus, $T$ is a solution of the \SMST problem. If \NSKA terminates with a failure, then from the second point it follows that there was no solution $T$, as otherwise \NSKA guesses a $T$-preferring universal order and terminates successfully.

The last three points are simple observations extending the facts \ref{fact:kruskal} into simultaneous setting using the previous two points.\qed
\end{proof}

\section{Cases and variants}

\begin{lemma}\label{lem:selfred}
Let $I$ be a feasible instance of the \SMST problem. Then any solution $T'$ of $I$ restricted to edges of weight at most $w$ can be extended to a solution $T$ of the whole $I$ by adding some edges of weight greater than $w$. Furthermore, this extension does not depend on $T'$. 
\end{lemma}
\begin{proof}
Let $T$ be a solution of the \SMST problem. We choose any $w$ and split the edges into a set of light edges $L$ of weight at most $w$ and a set of heavy edges $H$ with weight strictly greater than $w$. 

Consider running \SKA on the instance $I$ restricted to edges from $L$ using any $T$-preferring universal order. Since \SKA does not look ahead, it cannot distinguish whether it runs on a restricted instance or the full instance and therefore it does not fail and outputs $T$ restricted to $L$ (denoted $T[L]$), which is a solution of the restricted instance. Since $T'$ is also a solution of the restricted instance, both $T'$ and $T[L]$ define the same components on all individual graphs and have the same weight. Let us define $\bar{T} = T'\cup T[H]$. Clearly $\bar{T}$ is acyclic on each graph and has the same weight as $T$. Therefore $\bar{T}$ is a solution of the full instance, extending (any) $T'$. \qed
\end{proof}

\begin{observation}\label{cor:trivialstages}
Let us have an \SMST instance $I$ where $m(I)$ denotes the number of edges and $R(I)$ denotes the maximum number of repeats of any weight. If $R(I)! \in m(I)^{O(1)}$, in other words $R(I)$ is asymptotically very small, then $I$ can be solved in a polynomial time. 
\end{observation}
\begin{proof}
Suppose we implement the \NSKA deterministically and use backtracking to guess the next edge in the universal order. The previous lemma shows that it is sufficient to consider only backtracks within the current stage. If we ever need to backtrack into the previous stage, then the solution of the previous stage cannot be extended and therefore no solution exists. 

If all of the weights in our instance of the \SMST are either distinct, or the number of repeats of each value is asymptotically very small, then we can try all possible orders within each stage in polynomial time. More precisely whenever $R(I)! \in m(I)^{O(1)}$ we have at most polynomially many orderings in each stage and the algorithm finishes in polynomial time. If $R(I) \in \mathcal{O}(\log\log n)$ then there are at most linearly many possible orderings and the algorithm's running time differs by only a factor of $\mathcal{O}(m)$ from the \NSKA's running time on a non-deterministic machine. 
\qed
\end{proof}

\begin{definition}
A {\em Simultaneous $\{0,1\}$ Minimum Spanning Tree} problem, or \OISMST in short, is an instance of \SMST where we restrict all the edge weights to be either 0 or 1. 
\end{definition}

We show an equivalence of the general \SMST and \OISMST up to a polynomial factor of complexity. 
\begin{lemma}\label{lem:01reduction}
Any algorithm solving \OISMST in polynomial time can be used to solve general \SMST problem in polynomial time. 
\end{lemma}
\begin{proof} 
First let us consider an instance of \SMST using at most two distinct values for weights. Then we can replace these by 0 and 1. From the point of view of the individual graphs, each subset of edges is a minimum spanning tree after the modification if and only if the same holds before the modification; and so the same applies to the simultaneous minimum spanning trees. 

We continue via induction. Let us have an algorithm based on any \OISMST algorithm that solves any \SMST instance with at most $k$ distinct values of weight. We will extend this algorithm to $k+1$ values. Let us have an instance that uses $k+1$ values and let $w$ denote the highest one. We restrict $G$ to $G'$ by restricting to edges lighter than $w$. We already know how to solve \SMST for $G'$, acquiring a partial solution $T'$ or showing that no solution exists in which case the original \SMST has no solution. If we have the solution $T'$, then according to lemma \ref{lem:selfred} $T'$ can be extended by some edges of weight $w$ to a full solution. 

We once again modify $G$ into $\bar{G}$ as follows. We restrict $G$ to edges from $T'$ and edges of weight $w$. We set the weight of all edges from $T'$ to 0 and the weight of the remaining edges to 1. We now have an instance of \OISMST such that any solution contains all the edges from $T'$ as they form a partial simultaneous spanning tree and the \SKA would accept all of the edges regardless of the universal order used. Let $\bar{T}$ be a solution of the \OISMST problem on $\bar{G}$, then $\bar{T}$ is also a solution of the original \SMST problem and the algorithm outputs $\bar{T}$, otherwise we answer "no". 

To show completeness, suppose that there exists a solution $T$. Then we necessarily obtain $T'$ in the first step and $T'$ can be extended to a solution of the whole problem (not necessarily $T$) and thus the \OISMST on $\bar{G}$ has a solution $\bar{T}$. \qed
\end{proof}

\begin{definition}
An {\em Intersection-Heavy Simultaneous $\{0,1\}$ Minimum Spanning Tree} problem, or \POISMST in short, is an instance of \SMST where we restrict all the edge weights to be either 0 or 1. Furthermore all the edges of weight 1 are only in the intersection of all the individual graphs. 
\end{definition}

The motivation behind this restriction comes from a simple observation. 

\begin{observation}\label{obs:PIOSMST}
Let $I$ be an instance of \OISMST (for any number of graphs) where no edges of weight 1 appear in the intersection. Then after solving the first stage, the \SKA algorithm always finishes for any universal order $\pi$
\end{observation}
\begin{proof}
This is easy to see as each edge of weight 1 will be presented by the \SKA to a single instance of the Kruskal's algorithm and therefore in no step can the algorithm fail (get two opposite answers). Furthermore, one can see that the order of edges of weight 1 no longer matters, though different orders may give different solutions. \qed
\end{proof}

This observation formalizes an intuition that it is in some sense harder to deal with weight~1 edges in the intersection than in the exclusive parts. 

It might therefore seem that to solve a \OISMST problem, one might first greedily find a subset of edges from the intersection and then extend it to the exclusive parts. This approach fails on a simple example. Let us have exactly four vertices $a_1,a_2,b_1,b_2$ in the intersection. Let $G_1$ contain four weight~0 edges $a_1c_1, a_2c_2, b_1d_1, b_2d_2$, and let $G_2$ contain two weight~0 paths $P_i$ connecting $a_i$ and $b_i$ for both values of $i$. Finally let $a_1a_2$, $b_1b_2$ and $c_1c_2$ be weight~1 edges where the last one is exclusive for $G_1$. Clearly the only solution takes exactly the weight-1 edges $b_1b_2$ and $c_1c_2$. However if the graphs contains $d_1d_2$ rather than $c_1c_2$ then picking the edge $b_1b_2$ is not correct. Therefore an algorithm may not be oblivious to the exclusive parts. 
 
It seems logical to also consider the opposite approach, that is to first solve the exclusive parts where the solution seems rather fixed and then exploit the information from exclusive parts to extend the partial solution to the intersection. It is no surprise that this approach is flawed as well. As an example, let us have two graphs $G_1$ and $G_2$ where $G_1$ is only one edge $xy$ and $G_2$ is a triangle $xyz$. If we were to first find a maximum acyclic set of each exclusive part, we would get the subset $\{xz,yz\}$. However now we cannot extend this subset into a solution as there are only two solutions $\{xy,yz\}$ and $\{xy,xz\}$. 

Both of these greedy approaches to a \OISMST are flawed, even under the assumption that we are able to solve the first stage correctly in polynomial time. However according to the observation \ref{obs:PIOSMST} limiting all of the edges of weight 1 to the intersection gives instances that are in some sense easier, as the hardness of the problem is focused in the intersection which can be solved without considering exclusive weight~1 edges, as there are none. Later we show that \POISMST is actually equivalent to \OISMST, which will be a key step in solving the \OISMST problem. 

\begin{definition}
A {\em Simultaneous Spanning Tree} problem, \SST in short, is an unweighted version of the \SMST problem, in other words a \SMST problem using only one weight. 
\end{definition}

The \SST is clearly at most as hard problem as all of the previous versions of the \SMST and is an interesting problem on its own. We use the \SST as a simple base case in our construction later on. 

\begin{observation}
\SST $\subseteq$ \POISMST $\subseteq$ \OISMST $\subseteq$ \SMST 
\end{observation}
\section{Case $k \geq 3$ is \NP-complete}\label{sec:NP}

\begin{problem}[3D matching]
Let $U,V,W$ be disjoint finite sets such that $|U|=|V|=|W|=k$ and let $T$ be a subset of $U\times V\times W$. Is there a set $M\subset T$ with $|M|=k$, such that for any $x\in U\cup V\cup W$ there is exactly one hyperedge $e\in M$ such that $x\in e$.
\end{problem}
\begin{fact}[\cite{b:karp}]
3D matching is \NP-complete.
\end{fact}

\begin{theorem}
The problem of 3D matching can be polynomially reduced to \POISMST problem for 3 graphs. 
\end{theorem}
\begin{proof}
Without loss of generality we assume that every element of $U,V$ and $W$ is element of at least one hyperedge in $T$, otherwise the original 3D matching trivially has no solution. 

We define graphs $G_1,G_2,G_3$ and $H$ where $H = G_1 \cap G_2 \cap G_3$ forming a "sunflower" intersection, that is $H = G_i \cap G_j$ for each $i \neq j$. We associate $G_1$ with $U$, $G_2$ with $V$ and $G_3$ with $W$. 

First put a central vertex $c$ in $H$. For each hyperedge $e\in T$, put a vertex $v_e\in H$ and connect it to $c$ by an edge in $H$ of weight 1. For each element $x\in U$, put a vertex $v_x$ into the exclusive part of $G_1$ ($G_1 \backslash H$) and for every $e \in T$ such that $x \in e$, connect $v_e$ and $v_x$ by an edge of weight 0. Do the same for $V$ and $W$ with graphs $G_2$ and $G_3$ respectively. By construction these graphs form the required "sunflower" configuration. 

The structure of the graph $H$ can be alternatively described as follows. The intersection $H$ contains exactly a star with center $c$ and all edges of weight 1 where each ray represents a different element from $T$. 

Let us focus on $G_1$ and $U$, for the other graphs and sets the arguments are symmetrical. The graph $G_1$ is composed of the central star and exclusive vertices representing elements of the associated set $U$. Every vertex representing an element $x$ is connected via edges of weight 0 to all vertices representing the hyperedges that contain $x$. So for every element $x$, $v_x$ is a center of a weight-0 star in $G_1$. All of these weight-0 stars are disjoint as in each hyperedge there is at most one element from $U$. Since all the edges of weight 0 form an acyclic subgraph of $G_1$, every solution of this \SMST instance must contain all of them. Let $S$ be a solution of the \SMST problem. As for each $x \in U$, the $v_x$ is in the same component as $c$ in $G_1$, it must also be in the same component of $S[E(G_1)]$ and therefore at least one edge $cv_e \in S$ for some hyperedge $e$ such that $x \in e$. If it happened that $cv_f \in S$ for some other hyperedge $f$ with $x \in f$, then $cv_e,v_ev_x,v_xv_f,v_fc$ form a cycle in $G_1$ and we get a contradiction. 

This means that the hyperedges represented by the edges (where $e$ is represented by edge $cv_e$) of weight 1 in $S$ are a solution of the 3D-matching. This is true as each $x \in U$ belongs to exactly one of the hyperedges from $S$ and the same applies to every $y \in V$ and every $z \in W$. 

On the other hand, let $M$ be a solution of the 3D-matching. Then we can construct a solution of the \SMST by simply picking all the edges of weight 0 and all the edges of weight 1 that represent the hyperedgesedges from $M$. As previously, we observe that everything in $G_1$ is connected into a single component. If we only consider the edges of weigh 0 on the other hand, then for each $x,y \in U$ the vertices $v_x$ and $v_y$ are in distinct components and can only be connected via the central star. Therefore any solution must connect $G_1$ into a single component using at least $|U|$ edges of weight 1. Since $|U| = |M|$, the solution of the \SMST constructed from $M$ is clearly minimal. \qed
\end{proof}

\begin{corollary}
The problem \SMST and its variants \OISMST and \POISMST are \NP-complete for 3 and more graphs.
\end{corollary}
\section{Case $k = 2$ is in \PP}

In this section we show that the general \SMST problem is polynomially solvable. We progress via a tandem of reductions. We already know that the general \SMST can be solved using an algorithm for \OISMST for a cost of some polynomial factor. We further reduce instances of \OISMST to tasks that are more orderly and symmetrical in some sense. We then use this to reduce the task to \POISMST. Finally, we show that solving \POISMST can be reduced to a problem of intersection of two matroids, which is a polynomial problem for two graphs. As an intermediate step, we will also solve the \SST problem by reduction to a matroid intersection problem. 

\begin{definition}
Let $G_1$ and $G_2$ be two graphs intersecting in a common induced subgraph and let $F$ be a subset of edges of $G_1$ and $G_2$. We say that $F$ is {\em simultaneously acyclic} if $F$ restricted to each of the two graphs $G_1$ and $G_2$ forms an acyclic subgraph. 
\end{definition}

\subsection{Reduction of 2-graph \OISMST to 2-graph \POISMST}\label{sec:reductioncap01smst}

For technical reasons we first want to get rid of all edges of weight 1 that cross the boundary in between the intersection and one of the exclusive parts. 

\begin{observation}\label{obs:OInotransedges}
Every \OISMST instance can be transformed to an instance where all of the edges of weight 1 have either both ends in the intersection or both ends in an exclusive part of one graph. This transformation at most doubles the number of edges and vertices. 
\end{observation}
\begin{proof}
This can be achieved by a simple operation that shifts the edges into the exclusive parts. We take each edge $xy$ of weight 1 such that $x$ is in the intersection and $y$ in the exclusive part of one of the two graphs. We subdivide $xy$ into two edges $xz$ and $zy$ where the vertex $z$ lies in the exclusive part of the relevant graph. We set the weight of $xz$ to 0 and the weight of $zy$ to 1. Since the vertex $z$ has degree two, the 0-weight edge $xz$ is an element of each solution of the new \OISMST. It is now easy to see that we can construct the solution of the original \OISMST instance from any solution of the new instance by removing $xz$ and substitution of $zy$ with $xy$ (if it is part of the solution). \qed
\end{proof}

Another issue is that each of the two graphs may require a different number of edges of weight 1, while each edge from the intersection would increase the size of both solutions. 

\begin{observation}\label{obs:OIsamesizesolutions}
Every instance of \OISMST with two graphs $G_1$ and $G_2$ can be transformed into an instance where every minimum spanning tree of $G_1$ and every minimum spanning tree of $G_2$ contain the same number of edges of weight 1. This transformation at most doubles the number of edges and vertices. 
\end{observation}
\begin{proof}
Let $G$ denote the union of $G_1$ and $G_2$ and let $\bar{G}$ denote their intersection. From the properties of the \SKA (lemma \ref{lem:ska}) we know that we can determine beforehand the components of $G_1$ and $G_2$ after all the edges of weight 0 are processed and after all the edges of weight 1 are processed. We also know that in order to compute the restriction of the solution to the edges of weight 1 we do not need to know the exact choice of edges of weight 0, they are in fact independent. By considering the number of components of $G_1$ and $G_2$ just after processing all the edges of weight 0 and after processing all edges, we deduce how many edges of weight 1 must be added into the minimum spanning tree of each graph, which is equal to the difference of the two values. 

Suppose that the solution of \OISMST must contain $j_1$ edges of weight 1 from the graph $G_1$ and $j_2$ edges of weight 1 from the graph $G_2$. If $j_1 = j_2$ then we do not need to modify the instance, otherwise without loss of generality $j_1 > j_2$. We pick an arbitrary vertex $v$ from the exclusive part of $G_2$ and extend $G_2$ by $j_1-j_2$ leaves attached to $v$. All the leaves are new vertices and lie in the exclusive part of $G_2$; and all of the new edges have weight 1. Every spanning tree of $G_2$ must now contain all of these edges, while every solution of the original instance can be extended by exactly these edges. After this modification, $j_1 = j_2'$ where $j_2'$ denotes the new number of weight-one edges in the graph $G_2$ after modification. Note that this construction also works for the case  $j_2 = 0$, although this can be solved directly using \SKA. \qed
\end{proof}

\begin{lemma}\label{lem:OItoPOI}
The \OISMST problem for $k = 2$ is polynomially reducible to \POISMST problem for $k=2$ of asymptotically at most quadratic size. Furthermore if the set of edges of weight 0 of the original \OISMST instance is simultaneously acyclic, then the same is true for the new \POISMST instance. 
\end{lemma}
\begin{proof}
Let us have an instance of \OISMST and let $G_1,G_2$ denote the two graphs and $\bar{G}$ their intersection. Using the previous observation we can assume without loss of generality that all of the weight-1 edges have either both ends contained in $\bar{G}$ or both ends contained in the exclusive part of one of the two graphs; and that there exists a positive integer $j$ such that every solution of the \OISMST constrained to both $G_1$ or $G_2$ has exactly $j$ edges of weight 1. This increases the size of the problem by a small multiplicative constant. 

We modify the problem so that all of the edges from the exclusive parts are removed and equivalently modeled by gadgets that have edges of weight 1 only in $\bar{G}$. To do this, we consider all pairs $e=(e_1,e_2),f=(f_1,f_2)$ of edges of weight 1 such that $e$ is from the exclusive part of $G_1$ and $f$ is from the exclusive part of $G_2$. We create two new vertices $x^{ef}_1,x^{ef}_2$ in $\bar{G}$ and add edges $e_1x^{ef}_1,e_2x^{ef}_2,f_1x^{ef}_1,f_2x^{ef}_2$ of weight 0 and an edge $x^{ef}_1,x^{ef}_2$ of weight 1. After processing all pairs, we delete all the edges of weight 1 from the exclusive parts. 

Let $M$ be a solution of the modified instance of \OISMST (which is in fact \POISMST). First we observe that whenever $x^{ef}_1x^{ef}_2 \in M$ for some removed edges $e$ and $f$ then $x^{eg}_1x^{eg}_2 \notin M$ for any $g \neq f$ as otherwise $e_1,x^{ef}_1,x^{ef}_2,e_2,x^{eg}_2,x^{eg}_1,e_1$ forms a cycle in $M[G_1]$. To get a solution of the original instance, we remove all the extra edges of weight 0 and replace each edge $x^{ef}_1x^{ef}_2$ by edges $e$ and $f$. Let us denote the resulting set of edges $M'$. Consider the graph $G_1$ and the components defined by $M'$ restricted to $G_1$. It is easy to see that the components are the same as in $M$ with the exception of the new vertices which are now isolated. Also, the total weight of $M'$ restricted to each graph (of the original instance) is the same as the total weight of $M$ restricted to each graph (of the modified instance). We conclude that $M'$ is a minimum simultaneous spanning tree.

\begin{figure}
    \centering
    \includegraphics{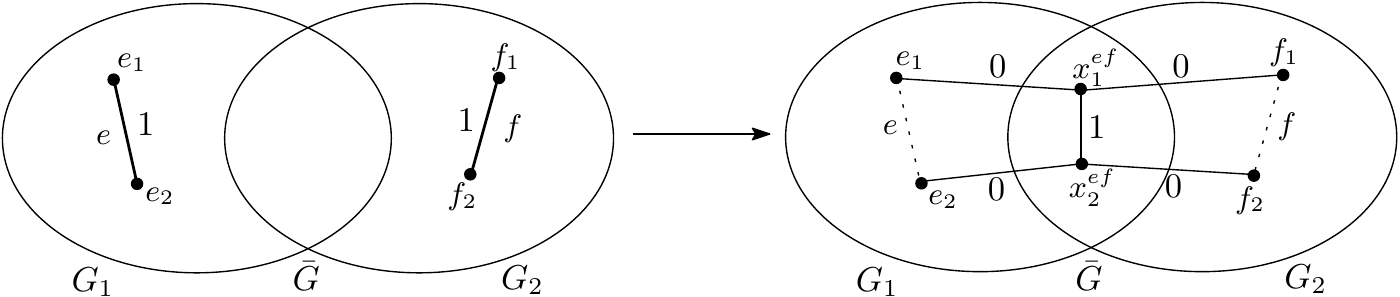}
    \caption{Gadget replacing pairs of edges}
    \label{fig:my_label}
\end{figure}

On the other hand, let $\bar{M}$ be a solution of the original instance. Since there is the same amount of edges of weight 1 in $\bar{M}$ restricted to $G_1$ and $G_2$, we can pair all of the edges from $\bar{M}$ of weight 1 that are in the exclusive parts of $G_1$ and $G_2$. We can now replace each pair of edges $e$ and $f$ by $x^{ef}_1x^{ef}_2$. After adding all the new edges of weight 0, we get a solution of the modified instance. Therefore the total cost of $\bar{M}$ is at most a total cost of the given solution. 

Supposing that the original edges of weight 0 form a simultaneously acyclic set, we observe that the same is true after the reduction, as each new cycle contains an edge of weight 1. Furthermore we added at most a constant number of edges and vertices for all of the pairs of original edges, obtaining a problem of asymptotically at most quadratic size compared to the input problem. \qed
\end{proof}

\subsection{Matroids}

\begin{definition}
A {\em matroid} $M$ is a pair $(E,I)$ where $E$ is a set of elements and $I$ is a family of independent sets (subsets of $E$) satisfying the following properties: 

\begin{enumerate}
    \item $\emptyset \in I$ 
    \item $\forall X,Y$ s.t. $ X \in I$ and $ Y \subset X$ : $Y \in I$
    \item $\forall X,Y \in I$ s.t. $|X| > |Y|$ : $\exists x \in X\setminus Y$ s.t. $Y \cup \{x\} \in I$
\end{enumerate}
\end{definition}

\begin{definition}
Let $G$ be a graph with a set of edges $E$ and $I$ be a set of all acyclic subsets of $E$. Then $(E,I)$ is a {\em graphic matroid} of $G$. 
\end{definition}

\begin{fact}
For any graph $G$ (possibly multigraph with loops), the graphic matroid of $G$ is a matroid and maximal independent sets of this matroid are exactly all possible spanning trees of $G$. 
\end{fact}

\begin{definition}
A {\em matroid intersection problem} of two matroids $(E,I_1)$ and $(E,I_2)$ on the same set of elements $E$ is the problem of finding a maximum subset of $E$ s.t. it is independent in both matroids. 
\end{definition}

\begin{fact}[\cite{b:edmonds}]
For a set $E$ and two matroids $(E,I_1)$ and $(E,I_2)$ given as independence oracles, the matroid intersection problem is solvable in polynomial time and polynomially many oracle queries. 
\end{fact}

\begin{fact}[\cite{b:graphicintersection}]
There are specialized algorithms for graphic matroid intersection problem. 
\end{fact}

\begin{lemma}\label{lem:grmatroid}
Let $G$ be a graph with edges divided into two disjoint subsets $F$ and $\bar{E}$ where $F$ is acyclic and $\bar{E} = E(G) \backslash F$. Let $I$ be a set of all subsets $X$ of $\bar{E}$ such that $F \cup X$ is an acyclic subgraph of $G$. Then $(\bar{E},I)$ is a graphic matroid. 
\end{lemma}
\begin{proof}
Let $H$ denote $G$ with all edges from $F$ contracted; we keep all the parallel edges and loops. We observe that the graphic matroid of $H$ is exactly $(\bar{E},I)$. \qed
\end{proof}

\subsection{Polynomiality}

\begin{theorem}\label{thm:SSTinP}
\SST $\in$ \PP for any number of graphs. 
\end{theorem}
\begin{proof}
To solve \SST, it suffices to use Kruskal's algorithm (or any other MST algorithm) to first take a minimum spanning tree of the intersection, and then extend this partial solution to each individual graph using only exclusive edges. Clearly each exclusive edge may only create a cycle in its respective graph. On the other hand we are never forced to take an exclusive edge closing a cycle (in fact, Kruskal's algorithm refuses such edges by definition). \qed
\end{proof}


\begin{lemma}\label{lem:acyclic01}
Let $I$ be an instance of \POISMST for two graphs such that the edges of weight 0 form a simultaneously acyclic set. Then $I$ can be solved in polynomial time using a matroid intersection algorithm. 
\end{lemma}
\begin{proof}
For each of the two graphs $G_i$ for $i \in \{1,2\}$ we define $F_i$ as the set containing all edges of weight 0 and $\bar{E}$ the set of all edges of weight 1. Let $I_i$ be a set of all subsets $X$ of $\bar{E}$ such that $X \cup F_i$ is acyclic in $G_i$ and let $M_i$ denote the pair $(\bar{E},I_i)$. According to Lemma \ref{lem:grmatroid} each $M_i$ is a matroid. Furthermore, both of the matroids are defined on the same ground set $\bar{E}$.

Let $F = F_1 \cup F_2$ be all the edges of weight 0. By lemma \ref{lem:selfred}, $F$ can be extended to a solution of the \POISMST by a suitable subset of $\bar{E}$. We can now use a (graphic) matroid intersection algorithm to find a set $X$ which is a maximum subset of $\bar{E}$ independent in both matroids $M_1$ and $M_2$. Therefore $X$ is the maximum subset of $\bar{E}$ that extends $F$ so that $X \cup F$ is simultaneously acyclic. If $X \cup F$ restricted to $G_1$ and $G_2$ spans all components, we output $X \cup F$, otherwise we answer "no". This is the same as to compare the size of $X \cup F$ to the size it should have.

Clearly if there exists a solution of the given \POISMST instance, then according to Lemma~\ref{lem:selfred} there exists a solution $Y$ extending the set $F$. The set of edges $Y \backslash F$ is an independent set in both matroids $M_1$ and $M_2$ and therefore $X$ exists and is of size $|Y \backslash F|$. This means that $X \cup F$ is a simultaneous spanning tree and the algorithm answers correctly. On the other hand, if no solution exists, then the set $X \cup F$ restricted to either $G_1$ or $G_2$ is acyclic but does not connect all the vertices connected in the original graph. We recognize this case and answer "no" correctly. \qed
\end{proof}

\begin{lemma}\label{lem:POIinP}
\POISMST $\in$ \PP for two graphs.
\end{lemma}
\begin{proof}
Let $I$ be an instance of the \POISMST problem. We show that we can solve $I$ using a (graphic) matroid intersection algorithm. 

First suppose that the edges of weight 0 are not simultaneously acyclic. We simply restrict $I$ to edges of weight 0, which gives us an instance of \SST. We can solve this instance in polynomial time according to Theorem \ref{thm:SSTinP}. If we obtain answer "no", then according to Lemma \ref{lem:selfred} there is no solution and we also answer "no". 

Suppose we get a solution $X$. Then, by Lemma \ref{lem:selfred}, we may delete all the edges of weight 0 except the edges from $X$ and further assume that the edges of weight 0 are simultaneously acyclic. We use Lemma~\ref{lem:acyclic01} to solve this reduced instance in polynomial time. \qed
\end{proof}

\begin{theorem}\label{thm:SMSTinP}
\SMST $\in$ \PP for two graphs.
\end{theorem}
\begin{proof}
Let us have an instance of the \SMST problem. According to Lemma \ref{lem:01reduction}, every instance of \SMST can be solved by solving at most $O(m)$ \OISMST problems, where $m$ denotes the number of edges on input. 

Any \OISMST can be polynomially reduced to \POISMST as shown in Lemma \ref{lem:OItoPOI}; and according to Lemma \ref{lem:POIinP}, each \POISMST instance can be solved in polynomial time. \qed
\end{proof}

\subsection{Complexity}

Let us have an instance of \SMST and let $n$ denote the number of vertices, $m$ the number of edges and $w$ the number of weights in the given instance. We proceed according to Theorem~\ref{thm:SMSTinP}. 

The \SMST problem is first decomposed into $(w-1)$ \OISMST subproblems. We observe that each edge in these subproblems is either already fixed as a part of the solution of \SMST or appears for the first time. The first kind of edges can be bound as at most $O(n)$ per \OISMST subproblem, as they must form a simultaneously acyclic set. The second kind can be bound as at most $O(m)$ over all of the \OISMST subproblems. 

Each of the \OISMST subproblems is reduced to a \POISMST problem of asymptotically at most quadratic size (by Lemma \ref{lem:OItoPOI}). Using the simultaneous acyclicity of edges of weight 0 we can use the approach of Lemma~\ref{lem:acyclic01} in all but the first subproblem, and use the Lemma~\ref{lem:POIinP} to solve the first subproblem. Therefore we solve at most $w$ (graphic) matroid intersection problems during the whole process and one instance of \SST problem. The final complexity depends on the choice of algorithms used to solve the matroid intersection problems and the \SST. 

Furthermore, if $w$ asymptotically approaches $m$, then some weight values have few representatives and more direct methods from Observation~\ref{cor:trivialstages} and Observation~\ref{obs:PIOSMST} using \SKA may be applied to reduce the complexity.

\section{Acknowledgements}

This paper is the output of the 2017 Problem~Seminar of Charles University. At this seminar undergraduate students attempt to solve open problems and learn to do research. We would like to thank Jan~Kratochv\'{i}l and Pavel~Valtr for their guidance, help and tea. 

The work was supported by the grant SVV--2017--260452 and grant CE-ITI P202/12/G061 of GA \v{C}R


\begin{thebibliography}{9}
\bibitem{b:sefe}
Thomas Bläsius, Stephen G. Kobourov, Ignaz Rutter. "Simultaneous Embedding of Planar Graphs" arxiv.org:1204.5853 or arxiv.org:1204.5853v3 (2015).

\bibitem{b:boruvka}
Borůvka, Otakar. "O jistém problému minimílním (About a certain minimal problem)", Práce mor. přírodověd. spol. v Brně III (3) (1926) 37–58 (Czech, German summary).

\bibitem{b:jarnik}
Jarník, Vojtěch: "O jistém problému minimálním", Práce Moravské Přírodovědecké Společnosti, 6, 1930, pp. 57-63.

\bibitem{b:prim}
Prim, Robert Clay. "Shortest connection networks and some generalizations." Bell Labs Technical Journal 36.6 (1957): 1389-1401.

\bibitem{b:kruskal}
Kruskal, Joseph B. "On the shortest spanning subtree of a graph and the traveling salesman problem." Proceedings of the American Mathematical society 7.1 (1956): 48-50.

\bibitem{b:mst}
Graham, Ronald L., and Pavol Hell. "On the history of the minimum spanning tree problem." Annals of the History of Computing 7.1 (1985): 43-57.

\bibitem{b:mstopt}
Pettie, Seth, and Vijaya Ramachandran. "An optimal minimum spanning tree algorithm." Journal of the ACM (JACM) 49.1 (2002): 16-34.

\bibitem{b:karp}
Karp, Richard M. "Reducibility among combinatorial problems." Complexity of computer computations. springer US, 1972. 85-103.

\bibitem{b:edmonds}
Edmonds, Jack. "Submodular functions, matroids, and certain polyhedra." Combinatorial Structures and Their Applications (Gordon and Breach, New York, 1970) 68-87. 

\bibitem{b:graphicintersection}
Gabow, Harold N., and Matthias Stallmann. "Efficient algorithms for graphic matroid intersection and parity." International Colloquium on Automata, Languages, and Programming. Springer Berlin Heidelberg, 1985.
\end{thebibliography}
\end{document}